\documentclass[12pt,psfig,reqno]{amsart}
\usepackage{amssymb,amsfonts,latexsym}
\usepackage{graphics,verbatim}
\usepackage{graphicx}
\usepackage[usenames]{color} 
\usepackage{soul} 

\hoffset=-1cm \errorcontextlines=0 \numberwithin{equation}{section}
\renewcommand{\baselinestretch}{1.1}
 \theoremstyle{plain}
\newtheorem{theorem}{Theorem}[section]

\newtheorem{proposition}[theorem]{\bf{Proposition}}

\renewcommand{\baselinestretch}{1.1}

\date {}

\begin{document}

\title[L2 norm preserving flow in matrix geometry]{L2 norm preserving flow in matrix geometry}

\author{Jiaojiao Li}

\address{Jiaojiao Li, Department of mathematics \\
Henan Normal university \\
Xinxiang, 453007 \\
China}

\email{lijiaojiao8219@163.com}

\thanks{ The research is partially supported by the National Natural Science
Foundation of China (No.)}

\begin{abstract}
In this paper, we study  L2 norm preserving heat flow in matrix geometry.
We show that this flow preserves the operator convex property and
enjoys the entropy stability in any finite time. Interesting
properties of this flow like conserved trace free property are also
derived.


\end{abstract}

\maketitle

\section{Introduction}
In this paper we continue our study of evolution equation in the
matrix geometry\cite{L}. We introduce the L2 norm preserving flow such that
starting from any initial data with unit L2 norm, we can produce a
family of matrices of unit L2 norm such that their limit is a
eigen-matrix of the Laplacian operator introduced in \cite{L} and
\cite{D}. In \cite{D}, the author introduces the Ricci flow which
exists globally when the initial data is a positive definite. The
Ricci flow preserves the trace of the initial matrix and the flow
converges the scalar matrix with the same trace as the initial
matrix. In \cite{L}, we have introduced the heat equation, which
also preserves the trace of the initial matrix. The advantage of the
heat equation is that the corresponding flow not only preserves the
trace but also can be with any initial matrix. In \cite{CL2} and
\cite{MC}, The authors introduce the norm preserving flows which are
global flow and converge to eigenfunctions. In quantum information
theory, we can see the interesting works \cite{Fei1}\cite{Fei2}\cite{Fei3} and \cite{Wu1}, 
but we need evolution flows which preserve the some norms like
the L2 norm just like the normalized Ricci flow \cite{H} preserves
the volume and the curve flows \cite{MC} \cite{MZ} preserve length
or area. This is the main topic of the study of norm preserving flow
in matrix geometry.

More precisely, we study again the matrix geometry model as in \cite{L} \cite{D}\cite{Ma}.
Let $X, Y$ be two Hermitian matrices on $C^{n}$ .
Define $U=e^{\frac{2\pi i}{n}X}$, $V=e^{\frac{2\pi i}{n}Y}$. We use
$M_{n}$ to denote the algebra of all $n\times n$ complex matrices
generated by $U$ and $V$ with the bracket $\{u,v\}=uv-vu$. Then
$CI$, which is the scalar multiples of the identity matrices $I$, is
the commutant of the operation $\{u,v\}$. Sometimes we simply use
$1$ to denote the $n\times n$ identity matrix.

We define two derivations $\delta_{1}$ and $\delta_{2}$ on the
algebra $M_{n}$ by the commutators
$$\delta_{1}:=[y,\cdot]\,\,\, \,\,\,  \delta_{2}:=-[x,\cdot]$$
Define the Laplacian operator on $M_{n}$ by
$$\hat{\Delta} =\delta_1^*\delta_1+\delta_2^*\delta_2=-\delta_{1}^{2}-\delta_{2}^{2}=-\delta_{\mu}\delta_{\mu},$$
where we have used the Einstein sum convention. We use the
Hilbert-Schmidt norm $|\cdot |$ defined by the inner product
$$<a,b>:=(a,b):=\tau(a^{*}b)$$
on the algebra $M_{n}$ and let $\sigma(a)=<1,a>$. Here $a^{*}$ is the complex conjugate of the
matrix $a$, $\tau$ denotes the usual trace function on $M_{n}$. We
now state basic properties of $\delta_{1}$, $\delta_{2}$ and
$\hat{\Delta}$ (see also \cite{D}).

For any $c\in M_n$, we define the Dirichlet energy
$$
D(c)=\sum <\delta_\mu c,\delta_{\mu} c>
$$
and the L2 mass
$$
M(c)=<c,c>.
$$
Let, for $c\not=0$,
$$
\lambda (c)=\frac{D(c)}{M(c)}.
$$
Then the eigenvalues of the Laplacian operator $\hat{\Delta}$ correspond
to the critical values of the Dirichlet energy $D(c)$ on the
constrain subset
$$
\Sigma=\{c\in M_n; M(c)=1\}.
$$
Hence to preserve the $L2$ norm, it is nature to study the evolution
equation
\begin{equation}\label{flow}
c_t=-\hat{\Delta} c+\lambda (c) c
\end{equation}
with its initial matrix $c|_{t=0}=c_0\in M_n$. Assume $c=c(t)$ is
the solution to the flow above. Then
$$
\frac{1}{2}\frac{d}{dt}M(c)=<c,c_t>=-<c,\hat{\Delta} c>+D(c).
$$
Since $<c,\hat{\Delta} c>=D(c)$, We know that $\frac{d}{dt}M(c)=0$. Then
$$
M(c(t))=M(c_0).
$$
The aim of this paper is to show that we always has a global flow to
(\ref{flow}) and the flow has many very nice properties like entropy
stability and operator convex preserving. Our main results, Theorem
\ref{global}, Theorem \ref{conver}, Theorem \ref{entropy}, and
Theorem \ref{convex} are contained in sections
\ref{sect1},\ref{sect2},\ref{sect3}, and \ref{sect4}.

\section{existence of the global flow}\label{sect1}

We first consider the local existence of the flow (\ref{flow}). We
prefer to follow the standard notation and we let $\Delta
=-\hat{\Delta}$. Let $a=a(t)\in M_n/CI$ be such that
\begin{equation}\label{eq1}
a_t=\Delta a+\lambda(t)a,
\end{equation}
with the initial data $a|_t=0 =a_0$. Here $a_0\in M_n/CI$ such that
$|a_0|^2=1$ and $\tau(a_0):=\bar{a}_0=0$. Then for $a=a(t)$, we let
\begin{equation}\label{lam}\lambda(t)=\frac{|\delta_\mu a|^2}{|a|^2}=-\frac{(\Delta
a,a)}{(a,a)}.\end{equation}

Formally, if the flow  (\ref{eq1}) exists, we then compute that
$$\frac{d}{dt}|a|^2=2(a,a_t)=2(a,\Delta a)+2\lambda(t)(a,a)=0.$$
Then $|a|^2(t)=|a|^2(0)=1,\forall t>0.$

Our main goal in this section is to show that it exists a global
solution to equation (\ref{eq1}) for any $a_0\in M_n/CI$ with
$|a_0|^2=1$.

Assume at first that $\lambda(t)\geq 0$ is any given continuous
function and $a=a(t)$ is the corresponding solution of (\ref{eq1}).
Define $b=e^{-\int \lambda(t)dt}a$. Then $b(0)=a(0)$ and we get

\begin{eqnarray} 
b_t &= & e^{-\int \lambda(t)d_t}(-\lambda)a+e^{-\int \lambda(t)d_t}a_t\\
&=& -\lambda ae^{-\int \lambda(t)d_t}+e^{-\int \lambda(t)dt}(\Delta a+\lambda a)\\
&=& e^{-\int \lambda(t)dt}\Delta a\\
&=& \Delta (e^{-\int \lambda(t)dt}a)\\
&=& \Delta b.\label{Delta b}
\end{eqnarray}

The equation (\ref{Delta b}) can be solved by standard way. For
completeness we recall it here. Assume $(\varphi_i)$ and
$(\lambda_i)$ are eigen-matrices and eigenvalues of $\Delta$, as we
introduced in \cite{L}, such that
$$-\Delta \varphi_i=\lambda_i\varphi_i,\,\,<\varphi_i, \varphi_j>=\delta_{ij}.$$
Note that $\lambda_i\geq 0.$

Assume that $b=b(t)$ is the solution to (\ref{Delta b}). Set
$$b=\sum<b,\varphi_i>\varphi_i:=\sum u_i\varphi_i,\,\,u_i\in
R,u_i=u_i(t).$$ Then by (\ref{Delta b}), we obtain
$$(u_i)_t\varphi_i=\Delta(u_i\varphi_i)=-u_i\lambda_i\varphi_i.$$
Then $(u_i)_t=-u_i\lambda_i$, and $u_i=u_i(0)e^{-\lambda_it}$.

Hence
\[
\begin{aligned}\label{eq4}
b=\sum u_i(0)e^{-\lambda_it}\varphi_i,
\end{aligned}
\]
and
\begin{equation}\label{eq5}
a=\sum u_i(0)e^{-\lambda_it+\int \lambda dt}\varphi_i
\end{equation}
solves (\ref{eq1}) with the given $\lambda(t)$.

We now define a iteration relation to solve (\ref{eq1}) for the
unknown $\lambda(t)$ given by (\ref{lam}).

Define $a_1$ such that $a_1$ solves the equation $a_t=\Delta
a+\lambda_0 a$ with $\lambda_0=-\frac{(\delta a_0,
a_0)}{(a_0,a_0)}$.

Let $k\geq 1$ be any integer. Define $a_{k+1}$ such that
\begin{equation}\label{eq6}
(a_{k+1})_t=\Delta a_{k+1}+\lambda_k(t)a_{k+1},\,\,\,a_{k+1}(0)=a_0,
\end{equation}
with
\begin{equation}\label{eq7} \lambda_k(t)=-\frac{(\Delta a_k,
a_k)}{(a_k,a_k)}=\frac{|\delta_\mu a_k|^2}{(a_k,a_k)}.
\end{equation}
Then using the formula (\ref{eq5}),we get a sequence $(a_k)$. We
claim that $(a_k)\subset M_n/CI$ is a bounded sequence and
$(\lambda_k(t))$ is also a bounded sequence. It is clear that
$(a_k)\subset M_n/CI$. If this claim is true, we may assume
$$a_k\rightarrow a,\,\, \lambda_k(t)\rightarrow \tilde{\lambda}(t).$$
Then by (\ref{eq6}) and (\ref{eq7}), we obtain
$$a_t=\Delta a+\tilde{\lambda}(t)a$$
and
$$\tilde{\lambda}(t)=-\frac{(\Delta a, a)}{(a,a)},$$
which is the same as (\ref{eq1}). That is to say, $a=a(t)$ obtained
above is the desired solution to (\ref{eq1}).

We first prove the Claim in a small interval $[0,T]$. Assume
$|a_k|\leq A=1.5$ and $|\lambda_k|\leq B=\log (4/3)$ on $[0,T]$,
$T=1/2$. Then, by(\ref{eq6}),
\begin{eqnarray}
\frac{1}{2}|a_{k+1}|^2_t &=&(a_{k+1}, (a_{k+1})_t)\\
&=&(a_{k+1}, \Delta a_{k+1})+\lambda_k|a_{k+1}|^2\\
&=&-|\delta a_{k+1}|^2+\lambda_k |a_{k+1}|^2.\label{eq8}
\end{eqnarray}
By (\ref{eq7}), we obtain $\lambda_{k+1}=\frac{|\delta_\mu
a_{k+1}|^2}{|a_{k+1}|^2}$. Then
$$|\delta_\mu a_{k+1}|^2=\lambda_{k+1}|a_{k+1}|^2.$$
By (\ref{eq8}), we get
\[
\begin{aligned}
\frac{1}{2}|a_{k+1}|^2_t &=-\lambda_{k+1}|a_{k+1}|^2+\lambda_k|a_{k+1}|^2\\
&=(\lambda_k-\lambda_{k+1})|a_{k+1}|^2.
\end{aligned}
\]
Then
$|a_{k+1}|^2=e^{2\int(\lambda_k-\lambda_{k+1})dt}|a_0|^2=e^{2Bt}\leq
A$.

Recall that $\bar{a}_0=0$. Then
$$\partial_t(\bar{a}_{k+1})=(\Delta a_{k+1}, 1)+\lambda_k(a_{k+1}, 1)=\lambda_k\bar{a}_{k+1},$$
so $$\bar{a}_{k+1}=e^{\int \lambda_k dt}\bar{a}_{k+1}(0)=e^{\int \lambda_k dt}a_0=0.$$

Note that since $\bar{a}_{k+1}=0$ and norm equivalence relation,
\[
\begin{aligned}
\lambda_{k+1}|a_{k+1}|^2 &=|\delta_\mu a_{k+1}|^2\\
&\leq C|a_{k+1}-\bar{a}_{k+1}|^2\\
&=C|a_{k+1}|^2.
\end{aligned}
\]
Then $\lambda_{k+1}\leq C$. Hence the Claim is true in $[0,T]$.

Therefore, (\ref{eq1}) has a solution in $[0,T]$. By iteration we
can get s solution in $[T,2T]$ with $u(T)$ as the initial data. We
can iterate this step on and on and we get a global solution to
(\ref{eq1}) with initial data $a_0$.

In conclusion we have the below.
\begin{theorem}\label{global}
For any given initial matrix $a_0\in M_n/CI$ with
$|a_0|^2=1$, The equation (\ref{eq1}) has a global solution with
$a_0$ as its initial data and $|a(t)|^2=1$ for all $t>0$.
\end{theorem}

\section{basic properties preserved by the flow}\label{sect2}

In this section we show that there are two properties of the initial
matrix are preserved except the conservation of the $L^2$ NORM.

 We
show that if the initial matrix is positive definite, then along the
flow (\ref{eq1}), the evolving matrix is also positive definite.
\begin{theorem}\label{positive}
Assume $a_0>0$, that is $a_0$ is a Hermitian positive definite.
Then $a(t)>0, \forall t>0$ along the flow equation
$$a_t=\Delta a+\lambda(t)a$$
with $a(0)=a_0$, where $\lambda(t)$ is given by (\ref{lam}).
\end{theorem}

\begin{proof} By continuity, we know that $a(t)>0$  for small $t>0$.
Compute
$$\frac{d}{dt}\log\det a=(a^{-1}, a_t)=(a^{-1}, \Delta a)+N\lambda(t)$$
where $N=\sigma(I)$.

Since
$$(a^{-1}, \Delta a)=-(\delta_\mu a^{-1}, \delta_\mu a)=(a^{-1}\delta_\mu a\cdot a^{-1}, \delta_\mu a)=|a^{-1}\delta_\mu a|^2.$$
We know that
$$\frac{d}{dt}\log\det a=|a^{-1}\delta_{\mu} a|^2+N\lambda(t)\geq N\lambda(t)\geq 0.$$
Hence, we have $a(t)>0$, $\forall t>0$.
\end{proof}
Remark that by continuity, we can show that if $a_0\geq 0$, then
$a(t)\geq 0$ along the flow (\ref{eq1}).

\begin{proposition}
Assume $\sigma(a_0)=\bar{a}_0=0$. Then $\bar{a}_t=0$, $\forall t>0$.
\end{proposition}

\begin{proof}
\[
\begin{aligned}
\frac{d}{dt}\bar{a} &=\sigma(a_t)\\
&=\sigma(\Delta a+\lambda(t)a)\\
&=\sigma(\Delta a)+\lambda(t)\sigma(a)\\
&=\lambda(t)\bar{a},
\end{aligned}
\]
so $\bar{a}(t)=\bar{a}(0)e^{\int^t_0\lambda(t)dt}=0$.
\end{proof}

\section{Convergence of the flow $a=a(t)$ at $\infty$}\label{sect3}
We prove the the global flow converges to some eigen-matrix.
\begin{theorem}\label{conver}
For any given initial matrix $a_0\in M_n/CI$ with
$|a_0|^2=1$, the global solution to the equation (\ref{eq1})  with
$a_0$ as its initial data converges to some eigen-matrix with its
eigenvalue $\lambda \geq \lambda_1$.
\end{theorem}

 Recall
that we have the global flow $a=a(t)$ such that
$$a_t=\Delta a+\lambda(t)a, \bar{a}=0$$
with $\lambda(t)=|\delta_{\mu} a|^2\geq C^{-1}|a|^2=C^{-1}$, $ a(t)|_{t=0}=a_0$, $|a_0|^2=1$ and $|a(t)|^2=1.$

Set
\[
\begin{aligned}
\Delta a=(\Delta a, a)a+(\Delta a)^\bot,
\end{aligned}
\]
where $((\Delta a)^\bot , a)=0$.

Note that $\lambda(t)=-(\Delta a, a)$, so by the Schwartz
inequality, $\lambda(t)\leq |\Delta a|$.

Compute
\[
\begin{aligned}
\frac{d}{dt}\lambda(t) &=2(\nabla a,\nabla a_t)\\
&=-2(\Delta a, a_t)\\
&=-2(\Delta a, \Delta a+\lambda(t)a)\\
&=-2(\Delta a)^2-2\lambda(t)(\Delta a, a)\\
&=-2(\Delta a)^2+2(\Delta a, a)^2\\
&=-2|(\Delta a)^\bot|\\
&\leq 0.
\end{aligned}
\]
Then we may assume that
$$a(t)\rightarrow a_\infty$$
and $|a_\infty|^2=1$, $\bar{a}_\infty=0$,
${\lambda}_\infty=\lim_{t\rightarrow \infty}\lambda(t)$ and $(\Delta
a_\infty)^\bot=0$.

The latter condition implies that
$$((\Delta a_\infty)^\bot, a_\infty)=0,$$
which is $-\Delta a_\infty=(\Delta a_\infty, a_\infty)a_\infty=\lambda_\infty a_\infty$.

Since $\lambda_\infty\geq C^{-1}$, we know that $\lambda_\infty$ is the non-zero eigenvalue of $-\Delta$
and then $\lambda_\infty\geq \lambda_1>0$.

This completes the proof of Theorem \ref{conver}.

\section{Entropy stability of the flow $a=a(t)$}\label{sect4}
We first prove the HS norm stability of the flow (\ref{eq1}). Recall
that there is a uniform constant $C>0$ such that for any $u,v\in
M_n$ and with $|u|^2=1=|v|^2$,
$$
|\lambda(u)-\lambda(v)|\leq C |u-v|.
$$
Given two initial matric $u_{0}$, $v_{0}$. Let $u, v$ be the
corresponding solutions to (\ref{eq1})) with initial datum
$u_0,v_0$.

Note that
\[
\begin{aligned}
\frac{d}{dt}|u-v|^{2} & = 2<u-v,\Delta u-\Delta v> +2<u-v,\lambda(u)u-\lambda(v)v>\\
&\leq 2<u-v, \Delta (u-v)>+2|\lambda(u)-\lambda(v)||u-v|+|\lambda(v)||u-v|^2\\
&\leq C_1|u-v|^{2},
\end{aligned}
\]
Where $C_1$ is a uniform constant.
$$|u-v|^{2}\leq e^{C_1t}|u-v|^{2}(0)$$
which implies the HS norm stability of the flow in (\ref{eq1}) in
the finite time interval $[0,T]$.

Remark: Similarly, we have the Trace norm stability of solutions,
where the Trace norm is denoted by $T(u,v)$ \cite{NC}.
Recall that by the spectral theorem we may let $u-v=Q-P$,
where $Q$ and $P$ are positive definite operators with compact support. Then
$$T(u,v)=\tau(P)+\tau(Q)$$

In below, we assume $u_{0}>0$ and define the von Neumann entropy
\cite{NC} by
$$S(u)=-\tau (u\log u)$$
for the positive definite solution $u=u(t)$ with $u(0)=u_{0}$.\\
Recall the Fannes inequality \cite{L} for $\forall a, b \in M_{n}$
and $a>0\,\, b>0$, we have
$$|S(a)-S(b)|\leq \Omega\log d +\eta (\Omega),$$
where $\eta (s)=-s\log s$, $d=dim M_{n}$ and $\Omega=\sum
|r_{i}-s_{i}|\leq T(a,b)\leq \frac{1}{e}.$

Then we can use the Fannes inequality to get the entropy stability
of the solution of $(\ref{eq1})$.
\begin{theorem}\label{entropy}
If $T(u_{0},v_{0})\leq\frac{1}{e}$, $u_{0}>0\,\, v_{0}>0$ in
$M_{n}$, then the solution $u(t), v(t)$ to (\ref{eq1}) satisfies
$$|S(u_{t})-S(v_{t})|\leq C_1T(u, v)(0)\log d+\eta (C_1 T(u,v))(0).$$
\end{theorem}

\begin{proof}
By the result above we have
$$T(u, v)(t)\leq C_1 T(u, v)(0),$$
by the Fannes inequality, we have
\[
\begin{aligned}
|S(u_{t})-S(v_{t})| & \leq T(u(t), v(t))\log d +\eta (T(u(t),v(t))\\
& \leq C_1T(u(0),v(0))\log d+C_1 \eta(T(u(0), v(0))),
\end{aligned}
\]
where we have used the monotonicity of the function $\eta$ in $[0,
\frac{1}{e}]$.
\end{proof}

\section{operator convexity of the heat equation}
In this section we prove SOME nonlinear convexity properties of the
initial matrix are preserved along the heat flow (\ref{heat}) below
\begin{equation}\label{heat}
a_t=\Delta a.
\end{equation}

Recall that we say $f:M_n\to M_n$ is \emph{operator convex} if for
any Hermitian symmetric matrices $A$ and $B$, we have
$$
\mu f(A)+(1-\mu)f(B)-f(\mu A+(1-\mu)B)\geq 0, \ \ \  \mu\in(0,1).
$$

Let $f(x)=x^2$. Then we compute,
\[
\begin{aligned}\label{eq11}
\Delta a^2 &=\delta_\mu(\delta_\mu a^2)\\
&=\delta_\mu(\delta_\mu a\cdot a+a\delta_\mu a)\\
&=\delta_\mu^2 a\cdot a+2(\delta_\mu a)^2+a\delta_\mu^2 a\\
&=\Delta a \cdot a+2(\delta_\mu a)^2+a \Delta a
\end{aligned}
\]
and
$$\partial_t(a^2)=aa_t+a_ta. $$
By $(\partial_t-\Delta)a=0$, we obtain
$(\partial_t-\Delta)a^2=-2(\delta_\mu a)^2$.

Hence, we have
\[
\begin{aligned}\label{eq12}
\frac{d}{dt}\log \det(a^2) &=\tau (a^{-2}(a^2)_t)\\
&=\tau(a^{-2}\Delta a^2)-2\tau(a^{-2}(\delta_\mu a)^2)\\
&=\tau(a^{-2}\Delta a^2)+2\tau(a^{-1}\delta_\mu a \cdot (\delta_\mu a)^\bot a^{-1})\\
&=\tau(a^{-2}\Delta a^2)+2\tau(a^{-1}\delta_\mu a\cdot (a^{-1}\delta_\mu a)^{\bot})\\
&\geq 0.
\end{aligned}
\]
We then deduce that $a^2>0$ provided $a_0^2>0$.

Assume that $\lambda>0$ is any positive constant. Recall that
$(\lambda+a)^{-1}>0$ if and only if $\lambda+a>0$.

Note that, by $\lambda+a_0>0$, we have $\lambda+a>0$.
Then $(\lambda+a)^{-1}>0$, $\forall t>0$.
\begin{theorem}\label{convex}
Assume $f$ is operator convex, if $f(a_0)>0$, then $f(a)>0$, $\forall t>0$.
\end{theorem}

\begin{proof}
According to \cite{Ki}, we have that any continuous operator convex
real function on $[0, \infty]$ can be expressed as
$$f(x)=f(0)+ax+bx^2+\int^\infty_0(\frac{x}{1+\lambda}-1+\frac{\lambda}{x+\lambda})d_\mu (\lambda),$$
where $a,b\geq 0$, $d_\mu$ is a nonnegative measure.

Since for $f(x)=x,x^2,\frac{1}{x+\lambda}$, we already know that if
$f(a_0)>0$, then $f(a)>0$, $\forall t>0$.

Hence, for any Hermitian positive definite matrix $a_0>0$ with $f(a_0)>0$, then
$$f(a)>0, \forall t>0.$$
This completes the proof of Theorem \ref{convex}

\end{proof}

\end{document}